\documentclass[10pt,twocolumn]{article}

\usepackage[letterpaper,margin=0.72in,columnsep=0.25in]{geometry}
\usepackage[T1]{fontenc}
\usepackage[utf8]{inputenc}
\usepackage{microtype}
\usepackage[hyphens]{url}
\usepackage{graphicx}
\usepackage{natbib}
\usepackage{caption}
\usepackage{algorithm}
\usepackage{algorithmic}
\usepackage{booktabs}
\usepackage{amsmath}
\usepackage{amssymb}
\usepackage{amsthm}
\usepackage{multirow}
\usepackage{listings}
\usepackage{xcolor}
\usepackage{array}
\usepackage{enumitem}
\usepackage{placeins}
\usepackage[hidelinks]{hyperref}

\newtheorem{lemma}{Lemma}
\newtheorem{proposition}{Proposition}
\lstdefinestyle{paperprompt}{
  basicstyle=\ttfamily\scriptsize,
  breaklines=true,
  breakatwhitespace=false,
  columns=fullflexible,
  keepspaces=true,
  showstringspaces=false,
  frame=single,
  rulecolor=\color{black!25},
  backgroundcolor=\color{black!2},
  xleftmargin=2pt,
  xrightmargin=2pt,
  aboveskip=4pt,
  belowskip=4pt
}
\lstdefinestyle{papersmall}{
  style=paperprompt,
  basicstyle=\ttfamily\tiny
}

\title{IR2Solve: Structured Intermediate Representations for\\
Cost-Efficient Optimization Autoformulation}
\author{%
  \textbf{Penglin Zhu} \quad
  \textbf{Linhai Zhang} \quad
  \textbf{Jungang Xu}\thanks{Corresponding author: \texttt{xujg@ucas.ac.cn}.} \quad
  \textbf{Xinchi Wei} \quad
  \textbf{Xiuqi Wu}\\[0.6em]
  School of Computer Science and Technology, University of the Chinese Academy of Sciences\\
  \texttt{\{zhupenglin21, zhanglinhai24, weixinchi22, wuxiuqi22\}@mails.ucas.ac.cn}\\
  \texttt{xujg@ucas.ac.cn}
}

\date{}

\hypersetup{
  pdftitle={IR2Solve: Structured Intermediate Representations for Cost-Efficient Optimization Autoformulation},
  pdfauthor={Penglin Zhu, Linhai Zhang, Jungang Xu, Xinchi Wei, Xiuqi Wu},
  pdfsubject={Large-language-model-based optimization autoformulation},
  pdfkeywords={optimization autoformulation, intermediate representation, large language models, deterministic verification}
}

\begin{document}
\maketitle

\begin{abstract}
Large language models (LLMs) can translate natural-language optimization
problems into solver-ready formulations, but direct code generation is brittle:
schema, indexing, and semantic errors can cause compilation failures,
infeasible models, or incorrect objectives, while iterative repair, search, and
multi-agent workflows increase inference cost. We present IR2Solve, an
intermediate-representation-first autoformulation pipeline that uses a single
semantic LLM call to produce a schema-constrained ModelIR, followed by two
deterministic stages: verification and IR-to-solver compilation. ModelIR
explicitly represents sets, parameters, variables, objectives, and constraints
using restricted Python-like expression strings. A concrete scalar-constraint
convention represents finite per-index constraint families as individual
entries, reducing free-index and implicit-quantification errors while
simplifying downstream verification and compilation. Across six cleaned
optimization benchmarks, IR2Solve achieves strong objective correctness and
remains competitive with recent optimization-modeling systems. A controlled
ablation on 153 IndustryOR and ComplexLP instances shows sequential gains from
the structured IR interface, the scalar-constraint instruction, and
deterministic verification. On a matched ten-instance cost panel, IR2Solve
uses one semantic call per instance, whereas Chain-of-Experts and SAC-Opt use
8 and 39 calls per instance and consume $3.3\times$ and $22.9\times$ the token
volume of IR2Solve, respectively. These results show that structured
intermediate representations, combined with deterministic post-generation
processing, provide a practical accuracy--cost trade-off for LLM-based
optimization autoformulation.
\end{abstract}

\section{Introduction}
\label{sec:intro}

Optimization models are central tools in operations research for complex
decision-making in healthcare, transportation, and scientific
discovery~\citep{delgado2022equity,jiao2024city,bartlett2017integer}.
The standard workflow has two stages: first, abstract the problem into a formal
optimization model; then, apply an algorithmic solver. While the second stage is
largely automated, the first---translating a real-world description into the
appropriate variables, objective, and constraints---remains a major bottleneck in
applied operations research.

Automatic modeling aims to reduce this burden by generating solver-ready models
from natural-language descriptions. Such interfaces can make optimization more
accessible to practitioners without specialized modeling training and allow
experienced modelers to spend less effort on routine formulation. Formally,
autoformulation maps an unstructured language description to a formal,
solver-executable model.

LLMs are attractive for this task because of their semantic and reasoning
capabilities~\citep{patil2024review,wei2022chain,yao2023tree,besta2024graph},
but their outputs also exhibit hallucinations, format inconsistencies, and
mathematical inaccuracies. These properties conflict with the precision and
logical consistency required by solver models. Existing systems improve
accuracy through prompt decomposition, multiple agents, reflection, search, or
iterative correction~\citep{xiao2023chain,ahmaditeshnizi2024optimus,astorga2024autoformulation,sacopt2026},
but this often increases API calls, token consumption, and workflow complexity.
The resulting accuracy--cost trade-off can hinder large-scale deployment.

To address this accuracy-cost trade-off, IR2Solve introduces an intermediate representation (IR) as a structured target
for LLM generation. A single semantic call maps the problem description into a
schema-constrained ModelIR. A deterministic verifier then applies a fixed set of
conservative rewrites, after which a deterministic compiler constructs the
Gurobi model. Mathematical expressions are stored as restricted Python-like
strings rather than complete programs. ModelIR also represents per-index
constraint families as concrete scalar entries, reducing free-index and
implicit-quantification errors and simplifying downstream checking. The final
pipeline uses one semantic LLM call and no iterative generation or LLM-based
rescue.

Our main contributions are:
\begin{enumerate}
\item We formulate optimization autoformulation through an explicit
representation layer and introduce ModelIR, which constrains open-ended code
generation into a structured, solver-oriented artifact.
\item We present IR2Solve, a practical one-call pipeline combining ModelIR with
a deterministic verifier and deterministic IR-to-solver compilation. A concrete
scalar-constraint convention is used to reduce index-scope ambiguity, while its
length and scaling limitations are stated explicitly.
\item We provide an evaluation on six cleaned benchmarks, a full
153-instance ablation of the structured IR interface, scalar-constraint
instruction, and verifier, and matched inference-cost profiling against
representative one-call, multi-agent, and iterative systems.
\end{enumerate}

The rest of the paper reviews related work, formalizes the task and its
challenges, presents IR2Solve, and evaluates accuracy, ablations, failure modes,
and inference cost. Full prompts, schema details, verifier rules, and
implementation notes are provided in the appendix.

\section{Related Work}
\label{sec:related}

\subsection{Code-First Autoformulation}
The most direct approach asks an LLM to generate complete solver code from a
problem description. Standard prompting and chain-of-thought provide natural
baselines~\citep{wei2022chain}; CAFA~\citep{deng24cafa} explicitly adopts coding
as the autoformulation interface and uses one generation followed by
deterministic cleaning. This design is inexpensive, but it couples mathematical
modeling with imports, solver APIs, indexing, status handling, and program
execution. IR2Solve instead makes the mathematical formulation an explicit
intermediate object before solver-specific construction. Our R0 condition is a
controlled one-call code-generation ablation, while CAFA is retained as a
stronger external code-first baseline.

\subsection{Multi-Agent and Reflective Repair}
Chain-of-Experts~\citep{xiao2023chain} decomposes formulation across
role-specialized agents coordinated by a conductor, with reflection. OptiMUS
\citep{ahmaditeshnizi2024optimus} organizes extraction, modeling, coding,
testing, and refinement in a modular workflow. SAC-Opt~\citep{sacopt2026}
reconstructs semantic anchors from generated code and iteratively corrects
misaligned objectives and constraints. These methods create additional
opportunities to revise errors, but require multiple calls and intermediate
interfaces. IR2Solve follows a non-iterative one-call generation path and uses
only deterministic verification after ModelIR is produced.

\subsection{Search and Execution-Aware Systems}
Autoformulation with MCTS~\citep{astorga2024autoformulation} treats modeling as
a search problem over many candidate formulations, using pruning and evaluation
to select a solution. NEMO~\citep{song2026nemo} represents a more recent
execution-aware direction with coding agents, sandboxed execution,
simulator--optimizer validation, memory, MBR decoding, and self-consistency.
Such systems demonstrate the value of inference-time scaling and solver-grounded
feedback. Our focus is complementary: we study how much reliability can be
obtained by changing the representation interface while retaining a one-call semantic workflow.

\subsection{Fine-Tuned OR Models and Benchmark Curation}
Fine-tuned systems such as ORLM~\citep{huang2025orlm}, LLMOPT
\citep{jiang2024llmopt}, and evolutionary or stepwise training approaches
\citep{wuevo} internalize optimization structure through curated data and
training. They can achieve strong accuracy but shift cost toward data synthesis
and model adaptation. Because noisy benchmark labels can materially change
conclusions, we use the cleaned releases and baseline values curated by the
recent survey~\citep{xiao2025survey}, covering NL4Opt
\citep{ramamonjison2023nl4opt}, MAMO EasyLP/ComplexLP
\citep{huang2024mamo}, NLP4LP from the OptiMUS lineage, IndustryOR, and
ReSocratic~\citep{yang2024optibench}.

\section{Problem Setup and Challenges}
\label{sec:problem}

\subsection{Optimization Modeling Preliminaries}
We consider linear and mixed-integer linear optimization problems over decision
variables $x$, instance data $\theta$, an objective, and inequality/equality
constraints:
\begin{equation}
\label{eq:opt_generic}
\begin{aligned}
\min_{x\in\mathcal X}\quad & f(x;\theta) \\
\text{s.t.}\quad & g_i(x;\theta)\le 0,\quad i=1,\ldots,I,\\
& h_j(x;\theta)=0,\quad j=1,\ldots,J,
\end{aligned}
\end{equation}
where $\mathcal X\subseteq\mathbb R^\ell\times\mathbb Z^k$ encodes continuous
and discrete variable domains. A formulation is a formal instance of
Eq.~\eqref{eq:opt_generic}.

\subsection{Task Setup and Representation Layer}
Let $d\in\mathcal D$ denote a natural-language description,
$m\in\mathcal M$ a mathematical formulation, and $c\in\mathcal C$ an
executable solver instance. Autoformulation can be factored as
\begin{equation}
 d \longrightarrow m \longrightarrow c.
 \label{eq:factor}
\end{equation}
The first map interprets the problem semantics; the second realizes the model in
a solver. Code-first systems collapse both into one LLM output. IR2Solve makes
$m$ an explicit artifact and implements the second map deterministically. A
useful stochastic view writes the two mappings as $p_\phi(m\mid d)$ and
$p_\psi(c\mid m)$, inducing
\begin{equation}
 p_{\phi,\psi}(m,c\mid d)=p_\psi(c\mid m)\,p_\phi(m\mid d).
 \label{eq:joint}
\end{equation}
This decomposition distinguishes modeling uncertainty from implementation
uncertainty. Direct code generation exposes both to one stochastic sample,
whereas IR2Solve reserves stochasticity for $d\to m$ and uses fixed programs
for $m\to c$.

An LLM does not directly emit an abstract $m$; it emits an artifact in a
representation language $L$. Let $\mathcal M_L\subseteq\mathcal M$ denote the
models expressible in $L$. The representation determines which artifacts are
admissible, what can be validated before execution, and how reliably they can be
compiled. Restricting generation to $\mathcal M_L$ can remove irrelevant
program structures and solver-API choices, but only under a coverage assumption:
the language must still express a correct formulation for the target instance.
We therefore use representation restriction as a design motivation, not as a
claim that ModelIR covers every optimization formalism. The current schema is
aimed at finite LP/MILP families represented in the benchmark suite; broader
coverage remains an empirical and engineering question. A conditional
probability statement and its proof are given in the appendix.

\subsection{Evaluation and Efficiency Metrics}
We report three execution-level outcomes. \textbf{Build} indicates that the
solver instance can be constructed without runtime errors. \textbf{Solved}
indicates that the solver returns a feasible or optimal result within the time
limit. \textbf{Objective correctness} requires a solved model whose objective $z$
matches the reference $z^\star$ under
$|z-z^\star|\leq 10^{-4}+10^{-6}|z^\star|$. This criterion is reproducible
but does not prove full semantic
equivalence: two different feasible regions can share the same optimum.

Efficiency is measured through semantic LLM calls, input/output/total tokens,
wall-clock latency, and API cost where available. Calls and tokens are the most
portable indicators across systems, while latency additionally depends on
provider and hardware conditions.

\subsection{Challenges}
Three challenges motivate the pipeline. \textbf{(C1) Large hypothesis space:}
natural-language descriptions admit many plausible variable definitions,
indexing choices, and constraint interpretations. \textbf{(C2) Fragile
executability:} even conceptually plausible models can fail because of schema,
shape, indexing, or solver-code errors. \textbf{(C3) Accuracy--cost trade-off:}
repeated sampling, reflection, search, and multi-agent coordination can improve
accuracy but increase calls and tokens. IR2Solve addresses these challenges by
constraining the generated artifact and reserving subsequent correction for a
small deterministic rule set.

\section{Methodology}
\label{sec:method}

\begin{figure*}[t]
\centering
\includegraphics[width=0.98\textwidth]{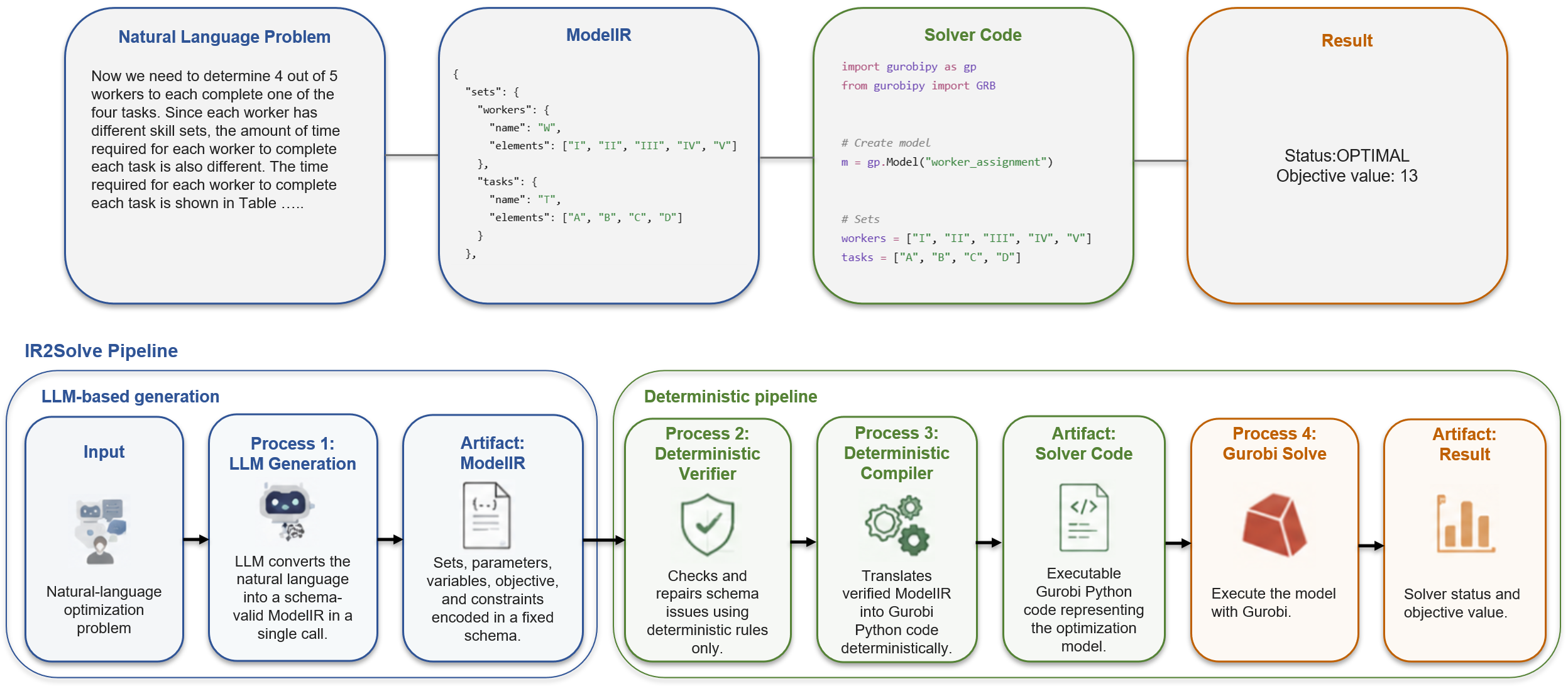}
\caption{IR2Solve separates one-call ModelIR generation from deterministic
verification, compilation, and solving. ModelIR explicitly represents sets,
parameters, variables, the objective, and scalar constraints. The verifier
applies fixed guarded rewrites, and the compiler constructs the Gurobi model
without additional semantic calls.}
\label{fig:pipeline}
\end{figure*}

\subsection{Overview}
Given a description $d$, IR2Solve first generates a formulation $m$ in ModelIR,
applies a deterministic verifier, and compiles the verified IR into a Gurobi
model $c$ (Figure~\ref{fig:pipeline}). This follows the factorization in
Eq.~\eqref{eq:factor}: the LLM is used for language-to-model interpretation,
whereas model-to-solver realization is deterministic. No candidate resampling,
reflection, critic, or semantic repair call is used in the final system.

\subsection{ModelIR}
ModelIR is a schema-constrained JSON representation with typed fields for
\texttt{sets}, \texttt{params}, \texttt{vars}, \texttt{objective}, and
\texttt{constraints}. Variables declare index sets, type in
$\{\text{continuous},\text{integer},\text{binary}\}$, and bounds. Objective and
constraint relations are stored as restricted Python-like expression strings:
mathematical strings containing arithmetic, declared parameter/variable
references, explicit indexing, and bounded \texttt{sum}/\texttt{quicksum}.
They contain no imports, model initialization, solver calls, control flow, or
result-handling code. During generation, the model is instructed to output JSON
only, use declared names consistently, and express the objective and constraints
through the permitted expression fragment. The full schema, prompt templates,
and a complete worked example are in the appendix.

This separation is deliberately modest: the LLM still performs the hard
semantic work of identifying variables, domains, objective, and constraints,
but it no longer has to synthesize the surrounding solver program. Schema
validation also makes failures observable before compilation rather than
surfacing only as runtime exceptions. Table~\ref{tab:modelir-schema} summarizes
these fields.

\begin{table}[t]
\centering
\scriptsize
\setlength{\tabcolsep}{3.1pt}
\begin{tabular}{p{0.23\columnwidth}p{0.66\columnwidth}}
\toprule
Field & Key contents \\
\midrule
\texttt{sets} & names and explicit finite elements \\
\texttt{params} & names, index signatures, and numeric values \\
\texttt{vars} & names, indices, variable types, and bounds \\
\texttt{objective} & minimization/maximization sense and expression string \\
\texttt{constraints} & named left expression, relation, and right expression \\
\bottomrule
\end{tabular}
\caption{Summary of the ModelIR schema.}
\label{tab:modelir-schema}
\end{table}

\paragraph{Constraint granularity.}
Each constraint entry is required to denote one concrete scalar relation. Free
indices and implicit ``for all'' semantics are disallowed, and a finite
per-index family is written as separate entries. For a finite set, this explicit
form is mathematically equivalent to the corresponding universally indexed
family. Operationally, the convention avoids index-scope inference during
compilation and supports direct constraint-level inspection. It can, however,
increase output length as indexed families grow, so we treat it as a practical
design choice for the targeted benchmark scale rather than a universally optimal
representation. The formal finite-set equivalence and its limitations are given
in the appendix.

\subsection{Deterministic Verifier}
The verifier applies a fixed sequence of conservative, guarded rewrites and
invokes no LLM. Its active functions can be summarized in three groups:
(i) canonicalize names, index keys, and common parameter shapes; (ii) resolve
recoverable index-scope and aggregation inconsistencies that would prevent
compilation; and (iii) apply a small set of unambiguous variable-domain and
constraint-direction sanity checks. Shallow syntax cases are handled by the
same sequence rather than exposed as independent modules. Verifier rule
families and implementation notes are provided in the appendix.

Let $T_1,\ldots,T_K$ denote the fixed guarded transformations. The verifier is
the deterministic composition
\begin{equation}
\widehat m=\mathcal V(m)=T_K\circ T_{K-1}\circ\cdots\circ T_1(m),
\label{eq:verifier-transform}
\end{equation}
where a transformation leaves its input unchanged when its guard is not
satisfied. Equation~\eqref{eq:verifier-transform} describes the implementation
directly: the verifier does not define a probability distribution, resample
candidates, or optimize a latent violation score. In the ablation, the verifier
is applied to the exact pre-verifier ModelIR used by the corresponding no-verifier
condition, so the comparison adds no model call.

\subsection{Deterministic Compilation}
The compiler instantiates declared sets, parameters, and variables; evaluates
restricted expressions in an environment containing only declared objects and
allowed aggregators; sets the objective; and emits solver constraints. Given the
same finalized IR and compiler version, the resulting Gurobi model is fixed.
This avoids a second stochastic IR-to-code stage and allows each emitted
constraint to be traced to its IR entry. Determinism does not guarantee that the
upstream model is semantically correct, but it preserves the finalized
representation, supports exact replay, and removes compiler sampling as an
additional source of variation. The current ablation does not isolate the
compiler against an LLM compiler; its contribution is therefore described as an
architectural property rather than an empirical accuracy delta.

\begin{algorithm}[t]
\caption{IR2Solve (one semantic call)}
\label{alg:ir2solve}
\textbf{Input}: description $d$ \\
\textbf{Output}: status and objective
\begin{algorithmic}[1]
\STATE $ir \gets \textsc{GenerateIR}(d)$ \hfill// one semantic call
\STATE $ir \gets \textsc{Verify}(ir)$ \hfill// deterministic
\STATE $c \gets \textsc{CompileToGurobi}(ir)$ \hfill// deterministic
\STATE \textbf{return} $\textsc{Solve}(c)$
\end{algorithmic}
\end{algorithm}

\section{Experiments}
\label{sec:exp}

We first report the main comparison on six cleaned benchmarks, then present a
controlled ablation on IndustryOR and ComplexLP, followed by cost and failure
analyses.

\subsection{Experimental Setup}
\label{sec:exp:setup}
\paragraph{Datasets.}
We use the survey-cleaned releases~\citep{xiao2025survey}: NL4Opt (214
instances)~\citep{ramamonjison2023nl4opt}, IndustryOR (42), EasyLP (545) and
ComplexLP (111) from MAMO~\citep{huang2024mamo}, NLP4LP (178) from the OptiMUS
lineage~\citep{ahmaditeshnizi2024optimus}, and ReSocratic (403)
\citep{yang2024optibench}. The ablation uses the complete IndustryOR and
ComplexLP sets because they contain frequent indexing, parameter-shape, and
semantic failures.

\paragraph{Configuration and metrics.}
The final local configuration uses GPT-4o-2024-08-06 with temperature 0 for a
single NL$\to$IR call and Gurobi with a 60-second per-instance limit. There is
no semantic retry, candidate selection, or LLM repair call. Reference objectives
are used only for retrospective evaluation. We report Build, Solved, and
objective-correct rates under the tolerances in Section~\ref{sec:problem}.
Baseline rows copied from other studies are marked by source and protocol.
Survey values use the same cleaned releases and report a standardized accuracy
metric, but the survey does not specify every decoding and solver parameter.
OptiMUS-0.3 and SAC-Opt values are five-run means from SAC-Opt's unified
evaluation. The appendix records the dataset, protocol, and
comparability details used in each table.

\subsection{Cleaned Benchmark Comparison}
\label{sec:exp:bench}
\begin{table*}[t]
\centering
\scriptsize
\setlength{\tabcolsep}{5.0pt}
\begin{tabular}{lcccccc}
\toprule
Methods & NL4Opt & IndustryOR & EasyLP & ComplexLP & NLP4LP & ReSocratic \\
\midrule
ORLM-LLaMA-3 8B$^{\ast}$ & 73.8\% & 42.9\% & 90.4\% & 59.5\% & 76.4\% & 61.8\% \\
Standard$^{\ast}$ & 61.2\% & 38.1\% & 70.3\% & 57.7\% & 73.6\% & 48.4\% \\
CoT$^{\ast}$ & 62.2\% & 40.5\% & 49.5\% & 42.3\% & 74.7\% & 43.6\% \\
Chain-of-Experts$^{\ast}$ & 66.7\% & 31.2\% & 94.4\% & 50.6\% & 87.4\% & 71.2\% \\
CAFA$^{\ast}$ & 68.1\% & 41.1\% & 71.2\% & 44.5\% & 50.0\% & 40.1\% \\
\midrule
OptiMUS-0.3$^{\dagger}$ & 79.8\% & 54.3\% & 92.4\% & 52.1\% & 89.8\% & 81.0\% \\
SAC-Opt$^{\dagger}$ & \textbf{86.8\%} & \underline{63.8\%} & \underline{96.5\%} & \textbf{79.6\%} & \textbf{94.0\%} & \textbf{88.7\%} \\
\midrule
\textbf{IR2Solve}$^{\ddagger}$ & \underline{86.4\%} & \textbf{64.3\%} & \textbf{97.4\%} & \underline{70.3\%} & \underline{90.4\%} & \underline{86.6\%} \\
\bottomrule
\end{tabular}
\caption{Objective-correct accuracy on the six survey-cleaned datasets.
$^{\ast}$Values reported by the cleaned-benchmark survey; ORLM is fine-tuned and
the other rows are prompting/agent baselines. $^{\dagger}$Five-run means
reported by SAC-Opt under its unified pipeline on the same cleaned dataset
releases; this protocol is not identical to the survey or our local evaluation.
$^{\ddagger}$IR2Solve uses the frozen final one-call configuration; four datasets
come from the 1,340-instance refresh, while IndustryOR and ComplexLP retain
matching final-configuration evaluations. Bold marks the best reported value
in each column, and underlining marks the second-highest value. The table is a
scoped clean-split comparison, not a fully controlled head-to-head across
backbones and workflows.}
\label{tab:benchmark}
\end{table*}

Table~\ref{tab:benchmark} reports the comparison on the six cleaned datasets.
IR2Solve is higher than the survey rows on all six columns. Relative to the
recent SAC-Opt report, IR2Solve has higher reported accuracy on IndustryOR and
EasyLP, while SAC-Opt is higher on NL4Opt, ComplexLP, NLP4LP, and ReSocratic.
This supports a scoped conclusion: a one-call IR-first system achieves
competitive accuracy on the cleaned splits without iterative semantic
correction. It is not a
fully controlled superiority claim because the survey, SAC-Opt/OptiMUS, and our
local evaluations differ in workflow details and run aggregation.

The table also clarifies the intended positioning. ORLM is a fine-tuned model,
CAFA a lightweight code-first method, Chain-of-Experts a multi-agent workflow,
OptiMUS a modular reflective workflow, and SAC-Opt an iterative
semantic-correction system. MCTS-based Autoformulation and NEMO are relevant
but are not inserted
into Table~\ref{tab:benchmark}: neither supplies the same six cleaned columns
under a directly alignable protocol. MCTS reports accuracy on original,
uncleaned benchmarks, whereas NEMO uses
autonomous coding agents, sandbox execution, memory, validation, and a broader
benchmark suite. We therefore discuss them as search-based and execution-aware
context rather than populate unsupported entries in the six-column comparison.

\subsection{Ablation Study}
\label{sec:exp:ablation}
\paragraph{Design.}
Four variants form a nested ablation on 153 instances. \textbf{R0} makes one
matched direct-code call. \textbf{R1} generates the final top-level ModelIR
structure but removes the concrete scalar-constraint instructions and uses no
verifier. \textbf{R2} adds those instructions, still without the verifier.
\textbf{R3} applies the deterministic verifier to the exact R2 ModelIR, so the
last contrast adds no model call. R0$\to$R1 evaluates the structured-IR
interface as a bundle (schema, expression strings, IR-oriented prompt, and
deterministic model construction); R1$\to$R2 evaluates the
constraint-granularity instruction; and R2$\to$R3 evaluates deterministic
verification. R0 is a controlled ablation, not the strongest possible
direct-code prompt; CAFA serves as the external code-first baseline in
Table~\ref{tab:benchmark}.

\begin{table*}[t]
\centering
\scriptsize
\setlength{\tabcolsep}{4.2pt}
\begin{tabular}{llcccccc}
\toprule
Variant & Design & \multicolumn{3}{c}{Objective correct (\%)} & Build (\%) & Solved (\%) & Struct. viol. \\
\cmidrule(lr){3-5}
 & & IndustryOR & ComplexLP & Pooled & Pooled & Pooled & Pooled \\
\midrule
R0 & Direct Gurobi code, one call & 7.1 & 7.2 & 7.2 & 37.3 & 35.3 & N/A \\
R1 & Structured expression-string IR & 50.0 & 30.6 & 35.9 & 44.4 & 42.5 & 522 \\
R2 & + scalar-constraint instruction & 64.3 & 55.0 & 57.5 & 71.2 & 68.6 & 39 \\
R3 & + deterministic verifier (final) & 64.3 & 70.3 & 68.6 & 79.1 & 75.2 & 12 \\
\midrule
R0$\to$R1 & Structured-IR stage gain & +42.9 pp & +23.4 pp & +28.8 pp & & & \\
R1$\to$R2 & Scalar-constraint stage gain & +14.3 pp & +24.3 pp & +21.6 pp & & & $-483$ \\
R2$\to$R3 & Verifier stage gain & +0.0 pp & +15.3 pp & +11.1 pp & & & $-27$ \\
\bottomrule
\end{tabular}
\caption{Ablation study on 153 cleaned instances. All objective, build, and
solved entries are percentages with one decimal. Structural violations count
free-index/non-scalar IR events and are not defined for the code-first R0
condition. Stage gains are observed sequential changes, not independent causal
contributions. The scalar-constraint and verifier contrasts are significant in
pooled paired tests ($p=1.1\times10^{-6}$ and $p=1.5\times10^{-5}$); the
aggregate R0$\to$R2 contrast is also significant ($p<10^{-19}$). The verifier
has no accuracy effect on IndustryOR but adds 15.3 points on ComplexLP.}
\label{tab:mechanism}
\end{table*}

\paragraph{Results.}
Table~\ref{tab:mechanism} shows pooled objective correctness increasing from
$7.2\%$ to $35.9\%$, $57.5\%$, and $68.6\%$. The structured IR interface gives
the largest observed sequential increase ($+28.8$ points pooled; $+33.2$
macro). The concrete scalar-constraint instruction adds $+21.6$ pooled points ($+19.3$
macro) and reduces free-index/non-scalar violations from 522 to 39. The verifier
adds $+11.1$ pooled points ($+7.7$ macro), with a dataset-dependent effect:
IndustryOR remains $64.3\%$, whereas ComplexLP increases from $55.0\%$ to
$70.3\%$.

The R1$\to$R2 and R2$\to$R3 contrasts have pooled exact McNemar
$p=1.1\times10^{-6}$ and $p=1.5\times10^{-5}$, with bootstrap 95\% intervals
$[13.7,29.4]$ and $[6.5,16.3]$ percentage points. The broader R0$\to$R2
IR-first contrast is also significant ($p<10^{-19}$). These are sequential
stage-wise effects, not independent causal contributions, and the schema,
expression language, prompt, and compiler were not separately isolated.

Build rate rises from $44.4\%$ in R1 to $71.2\%$ in R2, and the verifier further
raises it to $79.1\%$. Structural violations fall from 522 to 39 and then 12.
The results indicate that explicit constraint scope is useful on these datasets,
but they do not remove its scaling cost: large Cartesian index products can make
explicit generation substantially longer.

\paragraph{Executability diagnostics.}
Pooled compilation failures fall from 96 in R0 to 84, 42, and 30 in R1--R3.
Within the IR variants, undefined-name failures decline from 81 to 41 and 29.
R0 has no meaningful IR structural-violation count because it emits code rather
than ModelIR. Build and Solved remain above objective correctness for every
variant, showing that improved executability does not eliminate semantic errors.

\paragraph{Identification boundary.}
The nested design provides different levels of attribution. R1$\to$R2 and
R2$\to$R3 are narrow prompt/component contrasts, whereas R0$\to$R1 replaces
an entire interface: the output schema, expression representation, prompt
contract, and deterministic construction path change together. The experiment
therefore supports a bundle-level structured-IR effect, not independent claims
for schema, expression strings, or compilation. Likewise, the current design
does not compare deterministic compilation against an LLM compiler. This
boundary is important because a system can benefit from the combination even
when no single subchoice has been isolated.

\subsection{Inference Cost and Accuracy--Cost Analysis}
\label{sec:exp:cost}
\begin{table}[t]
\centering
\small
\setlength{\tabcolsep}{3.5pt}
\begin{tabular}{lccc}
\toprule
Method & Calls/inst. & Total tok. & Tok./inst. \\
\midrule
CAFA             & 1  & 7{,}954   & 795 \\
\textbf{IR2Solve}& 1 & 16{,}031 & 1{,}603 \\
Chain-of-Experts & 8  & 53{,}459  & 5{,}346 \\
SAC-Opt          & 39 & 367{,}028 & 36{,}703 \\
\bottomrule
\end{tabular}
\caption{Inference-cost profile on the COST-10 panel (10 instances, locally
measured, GPT-4o backbone). IR2Solve and CAFA use a single semantic call per
instance; Chain-of-Experts uses eight and SAC-Opt uses 39, spending
$3.3\times$ and $22.9\times$ more tokens than IR2Solve. COST-10 is a
\emph{cost}-profiling panel drawn from relatively simple NL4Opt-style cases,
selected to characterize inference cost rather than to expose the maximum
accuracy advantage of complex systems; we therefore do not use COST-10 accuracy
to rank modeling quality.}
\label{tab:cost}
\end{table}

\paragraph{Matched cost profile.}
COST-10 contains ten NL4Opt instances selected across source-length deciles and
is used only to profile inference cost. IR2Solve uses one call and 16,031 total
tokens. CAFA is cheaper at 7,954 tokens, also with one call. Chain-of-Experts
uses eight calls and 53,459 tokens; SAC-Opt uses 39 calls and 367,028
tokens. CoE and SAC-Opt therefore consume $3.3\times$ and $22.9\times$ as many
tokens as IR2Solve. COST-10 is not used to rank general modeling accuracy: its
instances are relatively simple and the panel is too small to expose the maximum
benefit of complex iterative workflows.

MCTS and NEMO are not included in this matched cost table because a faithful
COST-10 workflow could not be reproduced under the same instrumentation. MCTS
requires a search implementation and evaluator not available as a complete
cleaned-split pipeline, while NEMO depends on autonomous coding agents, sandbox
execution, memory, and inference-time selection. Their published workflows are
clearly more involved than one-call generation, but we do not infer token or
call totals from qualitative descriptions.

\paragraph{Cross-study proxy.}
Figure~\ref{fig:proxy} plots the six-benchmark means in
Table~\ref{tab:benchmark} against locally measured COST-10 tokens per instance.
CAFA is cheaper but less accurate; SAC-Opt reports a higher mean at far greater
inference cost; and IR2Solve lies between them with one semantic call. This is a
deployment proxy, not a strict Pareto proof, because the axes come from different
panels and protocols differ in backbone, extraction, and run aggregation.

\begin{figure}[h]
\centering
\includegraphics[width=0.94\columnwidth]{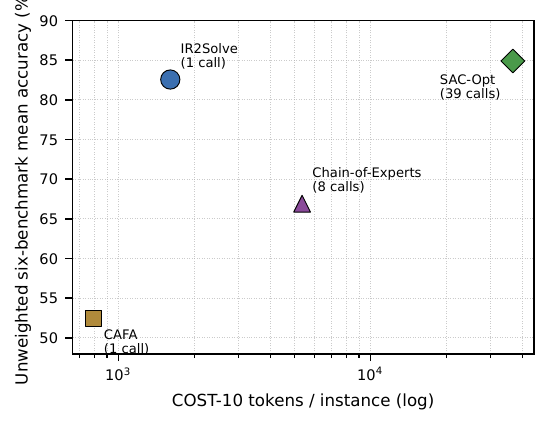}
\caption{Cross-study accuracy--cost proxy. Vertical: unweighted mean objective
accuracy across six cleaned datasets; horizontal: locally measured COST-10
tokens per instance (log scale), with calls annotated. Different source protocols make this
a deployment-oriented landscape rather than a same-instance Pareto test.}
\label{fig:proxy}
\end{figure}

\paragraph{Deployment interpretation.}
Calls and tokens describe complementary burdens. A one-call method is simpler
to schedule, audit, and retry, while token volume captures prompt and response
size within that call. CAFA minimizes both token use and orchestration overhead,
but its cleaned-benchmark mean is lower. SAC-Opt incurs considerably greater
inference cost to obtain a higher six-benchmark mean under its reported
protocol. IR2Solve occupies the intermediate operating point: one semantic
call, an explicit inspectable artifact, and a deterministic tail. The preferred
point therefore depends on the value assigned to accuracy, inference budget,
latency, and formulation auditability rather than on a single scalar score.

\subsection{Failure Analysis}
\label{sec:exp:failure}
Figure~\ref{fig:failure} decomposes pooled outcomes in the ablation. Direct code
generation is dominated by build failure. R1 converts many cases into executable
models, while the scalar-constraint instruction sharply reduces malformed
index-scope cases. The verifier moves additional cases into the
objective-correct segment. In R3, $20.9\%$ fail before build, $3.9\%$ build but
do not solve, and $6.6\%$ solve with an incorrect objective. These categories
are execution-level outcomes rather than manually adjudicated semantic causes.
Build failures can arise from unsupported or malformed representations, while
objective-incorrect cases may reflect missing constraints, incorrect domains,
or misread objectives. The deterministic tail can normalize known structures,
but it cannot reconstruct problem semantics that were never captured in the
IR. Objective correctness is also a coarse endpoint: a model can match the
reference optimum while differing in its feasible region.

\begin{figure}[t]
\centering
\includegraphics[width=0.94\columnwidth]{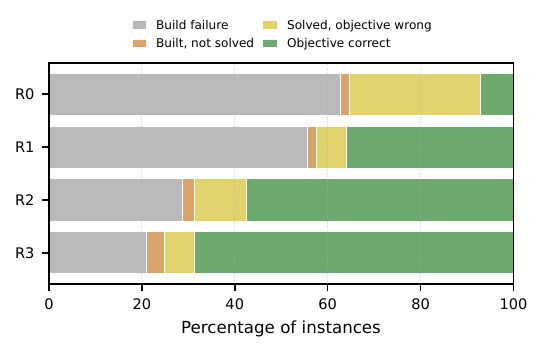}
\caption{Pooled failure conversion across the four ablation variants. Categories
are derived from nested Build, Solved, and objective-correct rates and sum to
$100\%$ within each row.}
\label{fig:failure}
\end{figure}

\section{Limitations}
\label{sec:limitations}

This study has several limitations. The mechanism ablation covers IndustryOR
and ComplexLP rather than every cleaned benchmark. The structured IR contrast
bundles schema, expression strings, prompt design, and deterministic model
construction, so their individual effects are not isolated. ModelIR is targeted
at finite LP/MILP-style formulations and does not claim universal coverage of
optimization formalisms. Explicit scalar expansion can grow rapidly with large
index products, increasing output length and truncation risk. The accuracy--cost
plot combines full-benchmark accuracy with cost measured on a separate
ten-instance panel, and is therefore a deployment proxy rather than a
same-instance Pareto proof. Finally, objective-value agreement does not establish
full semantic equivalence.

Addressing these limitations requires both broader representation support and
more direct evaluation of semantic coverage. Future work should extend ModelIR
while retaining deterministic compilation, develop source-grounded semantic
checks, and inspect recurring full-benchmark failure clusters. In particular,
representation coverage should be measured directly from reference formulations
rather than inferred from generation failures. Such a study should classify
reference models as representable, unsupported, or indeterminate and distinguish
native schema coverage from coverage obtained through mathematically equivalent
reformulation. This would sharpen the boundary between limitations of the
representation and those of the base language model.

\section{Conclusion}
\label{sec:conclusion}

IR2Solve separates natural-language interpretation from solver realization: one
semantic call produces ModelIR, after which a deterministic verifier and
IR-to-Gurobi compiler complete the pipeline. Across cleaned benchmarks, it
achieves competitive objective correctness with a one-call inference profile.
The ablation shows that the structured IR interface produces the largest
observed sequential increase, while the concrete scalar-constraint convention
and deterministic verification provide further, dataset-dependent improvements.

From a deployment perspective, the explicit IR offers a useful inspection point
between language understanding and solver execution. A practitioner can examine
declared variables, domains, objective direction, and individual constraints
before solving, while the deterministic tail makes the same artifact replayable.
This does not replace domain review, but it separates semantic modeling errors
from program-construction failures and provides a clearer basis for auditing than
an opaque end-to-end code sample. The one-call design also keeps operational
behavior predictable: cost does not depend on an open-ended number of reflection
or debugging rounds. Taken together, these results suggest that explicit
intermediate representations offer a practical middle ground between brittle
direct code generation and costly iterative autoformulation workflows while
improving inspectability and replayability.

\clearpage
\appendix
\setcounter{table}{0}
\setcounter{figure}{0}
\setcounter{algorithm}{0}
\renewcommand{\thetable}{A\arabic{table}}
\renewcommand{\thefigure}{A\arabic{figure}}
\renewcommand{\thealgorithm}{A\arabic{algorithm}}
\section{ModelIR Schema and Expression Language}
\label{app:schema}
A ModelIR instance has the top-level fields \texttt{meta}, \texttt{sets},
\texttt{params}, \texttt{vars}, \texttt{objective}, and
\texttt{constraints}. Table~\ref{tab:schema-detail} summarizes the dataclass
contract used by the final pipeline.

\begin{table*}[t]
\centering\small\setlength{\tabcolsep}{4pt}
\begin{tabular}{p{0.13\textwidth}p{0.13\textwidth}p{0.64\textwidth}}
\toprule
Field & Type & Recognized contents \\
\midrule
\texttt{meta} & object & objective sense (default \texttt{min}); nullable problem identifier, source, description, and version \\
\texttt{sets} & list & name, explicit finite elements, nullable description \\
\texttt{params} & list & name, zero/one/two-dimensional index signature, numeric values, nullable description \\
\texttt{vars} & list & name, index signature, type in \{continuous, integer, binary\}, bounds, nullable description \\
\texttt{objective} & object & name, min/max sense, one expression string, nullable description \\
\texttt{constraints} & list & name, left expression, relation in \{\texttt{<=}, \texttt{>=}, \texttt{==}\}, right expression, nullable description \\
\bottomrule
\end{tabular}
\caption{ModelIR field summary. Names referenced in expressions must be
declared in the same object. The prompt forbids extra keys. At parse time,
unsupported fields (e.g., \texttt{forall}, evidence, repair, or dependency
metadata) raise an error; other unrecognized fields are not represented by the
ModelIR dataclasses and therefore do not reach the compiler.}
\label{tab:schema-detail}
\end{table*}

Expressions admit numeric literals, arithmetic, parentheses, declared
parameters and variables, explicit indexing, and \texttt{sum}/\texttt{quicksum}
generators over declared finite sets. The compiler exposes only declared
objects plus \texttt{sum}, \texttt{quicksum}, \texttt{enumerate},
\texttt{range}, \texttt{len}, \texttt{max}, \texttt{min}, and
\texttt{abs}; Python built-ins are otherwise disabled.

\paragraph{Constraint granularity.}
Each constraint entry denotes one scalar relation. A finite per-index family is
expanded into multiple entries instead of carrying an implicit universal
quantifier. This reduces index-scope inference during compilation, but its
output length can grow with large index products.

\section{Complete Worked Example}
\label{app:example}
\noindent\textbf{Natural-language description.}\par\nobreak\smallskip
\begin{quote}\small
The Zhang family has six children: Harry, Hermione, Ron, Fred, George, and
Ginny. The respective costs of taking the six children are 1200, 1650, 750, 800,
800, and 1500. The family wants to minimize total cost, can take at most four children,
must take Ginny, cannot take Fred or George if Harry is taken, must take Fred
and Hermione if George is taken, and must take at least three children.
\end{quote}
\newpage
\noindent\textbf{ModelIR.}\par\nobreak\vspace{2pt}
\begin{lstlisting}[style=paperprompt]
{
  "meta": {
    "problem_id": "demo_family",
    "source": "worked_example",
    "description": "Family trip selection problem.",
    "sense": "min",
    "version": "v1"
  },
  "sets": [{
    "name": "Children",
    "elements": ["Harry", "Hermione", "Ron",
                 "Fred", "George", "Ginny"],
    "description": "Children considered for the trip."
  }],
  "params": [{
    "name": "cost",
    "indices": ["Children"],
    "values": {
      "Harry": 1200.0, "Hermione": 1650.0,
      "Ron": 750.0, "Fred": 800.0,
      "George": 800.0, "Ginny": 1500.0
    },
    "description": "Cost of taking each child."
  }],
  "vars": [{
    "name": "x",
    "indices": ["Children"],
    "vartype": "binary",
    "lb": 0.0,
    "ub": 1.0,
    "description": "Whether each child is taken."
  }],
  "objective": {
    "name": "minimize_cost",
    "sense": "min",
    "expr": "quicksum(cost[c] * x[c] for c in Children)",
    "description": "Minimize total trip cost."
  },
  "constraints": [
    {
      "name": "max_4_children",
      "expr_lhs": "quicksum(x[c] for c in Children)",
      "sense": "<=", "expr_rhs": "4.0",
      "description": "At most four children."
    },
    {
      "name": "take_ginny",
      "expr_lhs": "x['Ginny']",
      "sense": "==", "expr_rhs": "1.0",
      "description": "Ginny must be taken."
    },
    {
      "name": "harry_no_fred",
      "expr_lhs": "x['Harry'] + x['Fred']",
      "sense": "<=", "expr_rhs": "1.0",
      "description": "Harry and Fred cannot both be taken."
    },
    {
      "name": "harry_no_george",
      "expr_lhs": "x['Harry'] + x['George']",
      "sense": "<=", "expr_rhs": "1.0",
      "description": "Harry and George cannot both be taken."
    },
    {
      "name": "george_with_fred",
      "expr_lhs": "x['George'] - x['Fred']",
      "sense": "<=", "expr_rhs": "0.0",
      "description": "George implies Fred."
    },
    {
      "name": "george_with_hermione",
      "expr_lhs": "x['George'] - x['Hermione']",
      "sense": "<=", "expr_rhs": "0.0",
      "description": "George implies Hermione."
    },
    {
      "name": "at_least_3_children",
      "expr_lhs": "quicksum(x[c] for c in Children)",
      "sense": ">=", "expr_rhs": "3.0",
      "description": "At least three children."
    }
  ]
}
\end{lstlisting}
Each constraint is independently inspectable and compiles to one named Gurobi
constraint. The expression strings contain mathematical content but no solver
program structure.

\noindent\textbf{Verification, compilation, and solution.}\par\nobreak\smallskip
As written, the example already satisfies the schema and scalar-constraint
contract, so no verifier repair is required before compilation. The compiler
creates six binary variables and seven named constraints. Gurobi returns an
optimal solution selecting Ron, Fred, and Ginny, with objective value $3050$.
This completes the path from the natural-language description to ModelIR,
deterministic processing, and the solver result.

\section{Generation Prompts and Ablation Variants}
\label{app:prompts}

\subsection{R2/R3 ModelIR System Prompt}
The following is the exact system prompt. R2 uses it to generate ModelIR; R3
applies deterministic verification to the same R2-generated ModelIR and adds no
semantic call.
\begin{lstlisting}[style=paperprompt]
You are an expert in mathematical modeling and optimization.
Read a natural-language description of a linear / integer optimization problem and
output a JSON object describing the model in a strictly structured IR called ModelIR.

Hard requirements:
- Output JSON ONLY (no markdown, no explanations).
- Produce valid JSON (double quotes only; no trailing commas; no comments).
- Follow the schema exactly. Do NOT add extra keys.
- Use Python expressions in strings for objective/constraints
  (only arithmetic, indexing, quicksum/sum).
\end{lstlisting}

\subsection{R2/R3 ModelIR User Prompt Template}
The literal problem description is appended after the final line.
\begin{lstlisting}[style=papersmall]
Generate ONE JSON object with these top-level keys:
- meta
- sets
- params
- vars
- objective
- constraints

==== Schema ====
1) meta
{
  "problem_id": "string or null",
  "source": "string or null",
  "description": "string or null",
  "sense": "min" or "max",
  "version": "string or null"
}

2) sets
A list of sets. Each set is:
{
  "name": "string",
  "elements": ["string", ...],   // IMPORTANT: every element MUST be a STRING
  "description": "string or null"
}

3) params
A list of parameters. Each param is:
{
  "name": "string",
  "indices": [] | ["SetName"] | ["SetName1","SetName2"],
  "values": number | {key:number,...} | {i:{j:number,...},...},
  "description": "string or null"
}
Canonical forms for values:
- Scalar (0D): a number (a single-entry dict is tolerated by the compiler)
- 1D over set I: {"i1": 1.0, "i2": 2.0, ...}
- 2D over sets I,J (RECOMMENDED): {"i1": {"j1": 1.0, "j2": 3.0}, ...}

4) vars
A list of decision variables. Each var is:
{
  "name": "string",
  "indices": [] | ["SetName"] | ["SetName1","SetName2"],
  "vartype": "continuous" | "integer" | "binary",
  "lb": number,
  "ub": number or null,
  "description": "string or null"
}
IMPORTANT (integrality baseline):
- Default to "integer" for counts/units/number of items/visits/vehicles/facilities/assignments.
- Use "binary" only for yes/no decisions.
- Use "continuous" only for divisible amounts.

5) objective
{
  "name": "string",
  "sense": "min" or "max",
  "expr": "Python expression string",
  "description": "string or null"
}

6) constraints
A list of constraints. Each constraint is:
{
  "name": "string",
  "expr_lhs": "Python expression string",
  "sense": "<=" | ">=" | "==",
  "expr_rhs": "Python expression string",
  "description": "string or null"
}

==== Expression rules (compile-safe baseline) ====
Allowed building blocks inside expr strings:
- Numbers, + - * /, parentheses
- Indexing into params/vars: x[i], x[i][j], cost[i][j], demand[i]
- Aggregation: quicksum( ... for i in I ), sum( ... for i in I )
  (You may nest sums, e.g., quicksum(quicksum(x[i][j] for j in J) for i in I))

Hard constraints (must follow):
- Every constraint must be a SINGLE SCALAR constraint (lhs and rhs each evaluate to a scalar).
- Do NOT write implicit "for all" constraints inside one constraint.
  If a constraint must hold for each i, EXPAND it into multiple entries in the constraints list.
- Do NOT use free index symbols outside of a sum/generator that binds them.

Before outputting, self-check:
- JSON parses; keys match schema; no extra keys.
- All referenced names (sets/params/vars) are defined.
- No free indices; per-index constraints are expanded.

Now read the following optimization problem and output the JSON IR (JSON ONLY):
{PROBLEM_TEXT}
\end{lstlisting}

\subsection{R1 Scalar-Constraint Instruction Ablation}
R1 is derived mechanically from the R2 template: it removes exactly the
following block and the matching final self-check line; all other prompt text,
the schema, model, temperature, and compiler path are unchanged.
\begin{lstlisting}[style=paperprompt]
Hard constraints (must follow):
- Every constraint must be a SINGLE SCALAR constraint
  (lhs and rhs each evaluate to a scalar).
- Do NOT write implicit "for all" constraints inside one
  constraint. If a constraint must hold for each i, EXPAND it
  into multiple entries in the constraints list.
- Do NOT use free index symbols outside of a sum/generator
  that binds them.

Removed self-check:
- No free indices; per-index constraints are expanded.
\end{lstlisting}

\subsection{R0 Direct-Code Prompt}
R0 uses the existing code-generation prompt with \texttt{ir=None}. The
prompt was not specially optimized for this ablation, making R0 a controlled
one-call interface comparison rather than a strongest-possible code-first
baseline. Because the prompt accepts a ModelIR slot, the supplied IR value is
literally \texttt{null}.

\noindent\textbf{System prompt.}\par\nobreak\smallskip
\begin{lstlisting}[style=papersmall]
You are an expert in mathematical modeling and optimization.
Your task is to:
1. Read a natural language description of a linear/integer optimization problem.
2. Read the provided JSON immediate representation (ModelIR) of the modeling process of the problem.
3. Generate Python code that uses the Gurobi library to create the optimization model based on the ModelIR.

IMPORTANT: You MUST generate a Python function named `build_gurobi_model()` that:
1. Takes no parameters
2. Creates a Gurobi model using the provided information
3. Returns the created model

The generated Python code should follow these guidelines:
1. Import necessary libraries: import gurobipy as gp
2. Create the model: model = gp.Model("LP_Model")
3. Add variables according to variable information:
   - Continuous variables: model.addVar(lb=..., ub=..., vtype=gp.CONTINUOUS, name=...)
   - Integer variables: model.addVar(vtype=gp.INTEGER, name=...)
   - Binary variables: model.addVar(vtype=gp.BINARY, name=...)
4. Add constraints: model.addConstr(...)
5. Set objective function: model.setObjective(...)
6. Define a function `build_gurobi_model()` that creates and returns the model
7. At the end, call the function to build the model

Please output the complete Python code directly without any additional explanations.

EXAMPLE FORMAT:
```python
import gurobipy as gp
from gurobipy import GRB

def build_gurobi_model():
    # Create a new model
    model = gp.Model("LP_Model")
    
    # Your modeling code here...
    
    return model

# Build the model
model = build_gurobi_model()
\end{lstlisting}

\noindent\textbf{User prompt template.}\par\nobreak\smallskip
\begin{lstlisting}[style=papersmall]

SCHEMA FOR THE INTERMEDIATE REPRESENTATION (IR):
The IR is a structured representation of the problem, typically in JSON format.
It is not always correct. It includes the following components:

1) sets
A list of sets. Each set is:
{
  "name": "string",
  "elements": ["string", ...],   // IMPORTANT: every element MUST be a STRING (even if it looks like a number)
  "description": "string or null"
}

2) params
A list of parameters. Each param is:
{
  "name": "string",
  "indices": [] | ["SetName"] | ["SetName1","SetName2"],  // 0D/1D/2D only
  "values": number | {key:number,...} | {i:{j:number,...},...},   // recommended canonical forms below
  "description": "string or null"
}

Canonical forms for values:
- Scalar (0D): a number (or a single-entry dict is tolerated by the compiler)
- 1D over set I: {"i1": 1.0, "i2": 2.0, ...}
- 2D over sets I,J (RECOMMENDED): {"i1": {"j1": 1.0, "j2": 3.0}, "i2": {...}, ...}

3) vars
A list of decision variables. Each var is:
{
  "name": "string",
  "indices": [] | ["SetName"] | ["SetName1","SetName2"],  // 0D/1D/2D only
  "vartype": "continuous" | "integer" | "binary",
  "lb": number,
  "ub": number or null,
  "description": "string or null"
}

4) objective
{
  "name": "string",
  "sense": "min" or "max",
  "expr": "Python expression string",
  "description": "string or null"
}

5) constraints
A list of constraints. Each constraint is:
{
  "name": "string",
  "expr_lhs": "Python expression string",
  "sense": "<=" | ">=" | "==",
  "expr_rhs": "Python expression string",
  "description": "string or null"
}

TASK EXECUTION:
1. PROBLEM DESCRIPTION:
{PROBLEM_TEXT}

2. PROVIDED INTERMEDIATE REPRESENTATION (IR):
null

CRITICAL REQUIREMENTS:
- Generate a Python function named `build_gurobi_model()` that takes no parameters and returns the Gurobi model
- The function should use the provided ModelIR to build the model
- Include all necessary imports
- At the end of the code, call `build_gurobi_model()` to create the model
- Do NOT include any extra explanations outside the code

YOUR OUTPUT MUST BE A VALID python code that creates a Gurobi model.

\end{lstlisting}

\begin{table*}[t]
\centering\small\setlength{\tabcolsep}{4pt}
\begin{tabular}{lccccc}
\toprule
Variant & Generated artifact & Scalar instruction & Verifier & Semantic calls & Solver path \\
\midrule
R0 & Gurobi Python & N/A & none & 1 & generated code $\rightarrow$ Gurobi \\
R1 & ModelIR & removed & none & 1 & deterministic compiler \\
R2 & ModelIR & present & none & 1 & deterministic compiler \\
R3 & same R2 ModelIR & present & L1 then L2 & 0 additional & deterministic compiler \\
\bottomrule
\end{tabular}
\caption{Representation-ablation variants. All semantic generations use
GPT-4o-2024-08-06, temperature 0, one sample, and no semantic repair or second
candidate. R3 is paired with the exact R2 generation.}
\label{tab:variants}
\end{table*}

The prompt-template SHA-256 hashes are:
\begin{lstlisting}[style=papersmall]
R2 system:
b5a3ece68c765b9473a382696e63889d2c18f0ea4678ce872f5e1fdb773300d7
R2 user template:
1aa3ae328264cbe9657fc26624b367580a43b3ece754242e2c83a430cc72d898
R1 derived template:
b2b54cead5da346a479184c31da0b9ee3a9adbdcad376d8a26ec5a4e75bec130
\end{lstlisting}

\section{Deterministic Verifier Specification}
\label{app:verifier}
Rules execute once in the fixed order listed below.
Each rule first detects a pattern, mutates the current IR if its guard holds, and
passes that updated IR to the next rule. There is no LLM, resampling, score,
acceptance gate, or data-dependent stopping rule.

\begin{enumerate}[leftmargin=*,label=\textbf{\arabic*.},itemsep=3pt]
\item \textbf{Set/key canonicalization.}
Guard: any set element is non-string, or a 1D/nested-2D parameter key is
non-string. Rewrite: apply \texttt{str()} and keep the first value if two keys
collapse to the same string. If any declared set consists entirely of digit
strings, every numeric bracket token such as \texttt{[123]} in objective and
constraint expressions becomes \texttt{['123']}.

\item \textbf{2D parameter normalization.}
Guard: a parameter has exactly two indices, dictionary values that are not
already a dict-of-dicts, and at least one key parsable as \texttt{(i,j)},
\texttt{i,j}, or \texttt{i|j}. Rewrite: convert parsable entries to
\texttt{\{i:\{j:value\}\}} form. Unparsable entries are not inserted, and the
rule is a no-op if no key parses.

\item \textbf{Missing square diagonal.}
Guard: a nonempty 2D parameter is indexed by the same declared set twice, its
value is a nonempty dict-of-dicts, and a row or diagonal cell is missing.
Rewrite: create missing rows and set missing diagonal cells to \texttt{0.0}.
The guard does not inspect source text; the zero-diagonal assumption is a
deterministic benchmark heuristic, not a universal semantic guarantee.

\item \textbf{Free-index unrolling.}
Guard: an AST load name is in
\{\texttt{i,j,k,t,u,v,n,m,p,q,w}\}, is not generator-bound, and is neither a
declared object nor an allowed global. Its set is inferred from an explicit
\texttt{for s in Set} cue or indexed-object occurrence counts; under competing
candidates the leader must exceed the runner-up by at least two occurrences.
The selected set must contain 1--50 elements. Rewrite: replace one free symbol
by each set element and emit one named scalar constraint per element.

\item \textbf{Sum-call normalization.}
Guard: comma-before-generator syntax such as
\texttt{quicksum(x[i], for i in I)}, or an AST-valid top-level
\texttt{sum}/\texttt{quicksum} call with multiple positional arguments and no
generator/keywords. Rewrite: remove the comma before \texttt{for}, or replace
the multiple-argument call by a parenthesized additive chain.

\item \textbf{Integrality sanity.}
Guard: a variable is continuous and a substring cue occurs in its generated
name/description or indexed set name/description. Variable cues are
assignment/select/choose/take/open/close/build/install/use/activate/visit/route/
serve/facility/worker/task/item/node/edge/arc; set cues additionally include
city and job. Rewrite: bounds $[0,1]$ become binary; all other detected
variables become integer, with binary bounds normalized to 0 and 1.

\item \textbf{Direction sanity.}
Guard: the relation maps to equality and the generated name/description
contains an at-least cue
(\texttt{at least}, \texttt{no less}, \texttt{minimum}, \texttt{min },
\texttt{>= }, \texttt{atleast}) or an at-most cue
(\texttt{at most}, \texttt{no more}, \texttt{maximum}, \texttt{max },
\texttt{<= }, \texttt{atmost}, \texttt{up to}, \texttt{upto},
\texttt{limit}). Rewrite: equality becomes $\geq$ or $\leq$; if both cue
families occur, the rule is a no-op.
\end{enumerate}

Rules 6--7 inspect metadata generated in the same semantic call, not
independently grounded source spans. The ablation measures their observed net
effect and does not certify every individual rewrite.

\begin{table}[t]
\centering\scriptsize\setlength{\tabcolsep}{3pt}
\begin{tabular}{p{0.25\columnwidth}p{0.31\columnwidth}p{0.31\columnwidth}}
\toprule
Rule & Before & After \\
\midrule
Set elements & \texttt{I=[1,2,3]} & \texttt{I=["1","2","3"]} \\
2D parameter & \texttt{\{"a,b":5\}} & \texttt{\{"a":\{"b":5\}\}} \\
Square diagonal & off-diagonal $d[a,b]$ only & add $d[a,a]=d[b,b]=0$ \\
Free index & \texttt{x[i]>=2}, $I=\{a,b,c\}$ & three scalar entries \\
Sum syntax & \texttt{quicksum(x[i], for i in I)} & valid generator form \\
Integrality & continuous \texttt{assign\_x}, $[0,1]$ & binary \\
Direction & equality described ``at most 5'' & \texttt{<=} \\
\bottomrule
\end{tabular}
\caption{Unit-tested before/after examples for the seven active rules.}
\label{tab:verifier-examples}
\end{table}

\section{Parser, Compiler, and Execution Contract}
\label{app:compiler}

\paragraph{Parsing.}
The runner extracts the first fenced JSON object, otherwise the first balanced
JSON object, and then parses it. The ModelIR dataclasses contain only the
fields in Table~\ref{tab:schema-detail}. Unsupported fields such as
\texttt{forall}, evidence, annotations, dependencies, repair metadata, or
constraint-family metadata raise a parse error.

\paragraph{Compilation.}
The compiler creates explicit set lists; scalar, 1D, or nested-2D parameters;
and scalar, 1D, or 2D Gurobi variables. Dimensions above two raise
\texttt{NotImplementedError}. It evaluates the objective and each scalar
constraint in the restricted environment described in Section~\ref{app:schema},
adds one named Gurobi constraint per IR entry, and calls \texttt{model.update()}.
The R3 execution path is:
\begin{lstlisting}[style=paperprompt]
extract_json_from_text -> json_to_model_ir -> run_l1
-> run_l2 -> ir_to_gurobi -> optimize
\end{lstlisting}

\paragraph{Build and solve outcomes.}
A case is built if compilation returns a Gurobi model without an exception.
The runner sets \texttt{OutputFlag=0} and \texttt{TimeLimit=60}, then calls
\texttt{optimize}. It records Gurobi statuses
\texttt{OPTIMAL}, \texttt{INFEASIBLE}, \texttt{INF\_OR\_UNBD},
\texttt{UNBOUNDED}, \texttt{TIME\_LIMIT}, and \texttt{SUBOPTIMAL} when
available. The objective is read only when \texttt{SolCount>0}; thus the
execution-level \emph{Solved} count is exactly the number of cases with a
recorded objective, including a feasible incumbent under a non-optimal status.
A compilation exception is stored as \texttt{build\_error}; an exception during
optimization is stored as \texttt{solve\_error:<ExceptionType>}.

\paragraph{Objective correctness.}
The executable comparator is
\[
 |z-z^\star|\leq 10^{-4}+10^{-6}\max(1,|z^\star|).
\]
This is the standard floor-at-one interpretation of
\texttt{is\_close(atol=1e-4, rtol=1e-6)} used in the local evaluation
implementation. Reference objectives are loaded only for evaluation scoring,
after generation and solver execution.

\section{Formal Notes and Representation Boundaries}
\label{app:formal}
\subsection{Conditional Representation Restriction}
Let $\mathcal M_L$ be the formulations expressible in representation $L$ and
$\mathcal M^\star(d)$ the correct formulations for description $d$.
\begin{lemma}[Conditional restriction]
If $\mathcal M^\star(d)\subseteq\mathcal M_L$ and
$\Pr[m\in\mathcal M_L\mid d]>0$, then
\[
\Pr[m\in\mathcal M^\star(d)\mid m\in\mathcal M_L,d]
\geq \Pr[m\in\mathcal M^\star(d)\mid d].
\]
\end{lemma}
\begin{proof}
The coverage assumption makes the joint event of correctness and
representability equal to the correctness event. Dividing by a conditioning
probability in $(0,1]$ yields the inequality.
\end{proof}
This conditional statement does not show that ModelIR covers every correct
formulation or that schema-constrained prompting samples from the corresponding
conditional distribution.

\subsection{Finite-Instance Scalarization}
\begin{proposition}[Finite-instance scalarization]
For a finite set $I=\{i_1,\ldots,i_m\}$, the family
$g_i(x)\leq0$ for every $i\in I$ and the explicit conjunction
$g_{i_1}(x)\leq0,\ldots,g_{i_m}(x)\leq0$ define the same feasible region.
\end{proposition}
\begin{proof}
Universal quantification over a finite set is exactly the conjunction of its
finitely many instantiated statements.
\end{proof}
The proposition establishes equivalence of the two notations; it does not imply
that explicit expansion is shorter or preferable at arbitrary scale.

\subsection{Coverage Boundary}
The schema targets finite LP/MILP-style formulations with numeric
parameters and at most two-dimensional indexed objects. It does not claim native
coverage of general nonlinear, stochastic, robust, dynamic, conic,
semidefinite, infinite-dimensional, or continuously indexed models. Equivalent
reformulations, when available, require instance-specific modeling.

\section{Experimental Protocol and Paired Statistics}
\label{app:protocol}
The 153-instance ablation contains all 42 IndustryOR and 111 ComplexLP cases.
R0--R2 each use one GPT-4o-2024-08-06 call at temperature 0; R3 reuses the
paired R2 output. There is no semantic retry except an identical-request retry
on transient transport failure, no provider fallback, no candidate selection,
and no reference objective in any prompt. Gurobi uses the 60-second limit.

For a contrast, let $n_{00}$ be incorrect under both variants, $n_{10}$ correct
only under the left variant, $n_{01}$ correct only under the right variant, and
$n_{11}$ correct under both. The reported McNemar test is the exact two-sided
binomial test on $(n_{10},n_{01})$. Confidence intervals are percentile paired
bootstrap intervals for the right-minus-left accuracy difference, resampling
instances with replacement 10,000 times using seed 20260726. Pooled effects
weight every case equally; macro effects give IndustryOR and ComplexLP equal
weight.

\begin{table}[t]
\centering\tiny\setlength{\tabcolsep}{2.5pt}
\begin{tabular}{llrrr}
\toprule
Scope & Contrast & $(n_{00},n_{10},n_{01},n_{11})$ & $p$ & CI (pp) \\
\midrule
Pooled & R0$\to$R2 & (62,3,80,8) & $2\mathrm{e}{-20}$ & [41.8,58.8] \\
Pooled & R1$\to$R2 & (58,7,40,48) & $1.1\mathrm{e}{-6}$ & [13.7,29.4] \\
Pooled & R2$\to$R3 & (48,0,17,88) & $1.5\mathrm{e}{-5}$ & [6.5,16.3] \\
IndustryOR & R0$\to$R2 & (15,0,24,3) & $1.2\mathrm{e}{-7}$ & [42.9,71.4] \\
IndustryOR & R1$\to$R2 & (15,0,6,21) & $3.1\mathrm{e}{-2}$ & [4.8,26.2] \\
IndustryOR & R2$\to$R3 & (15,0,0,27) & 1.0 & [0.0,0.0] \\
ComplexLP & R0$\to$R2 & (47,3,56,5) & $1.2\mathrm{e}{-13}$ & [37.8,57.7] \\
ComplexLP & R1$\to$R2 & (43,7,34,27) & $2.5\mathrm{e}{-5}$ & [13.5,35.1] \\
ComplexLP & R2$\to$R3 & (33,0,17,61) & $1.5\mathrm{e}{-5}$ & [9.0,22.5] \\
\bottomrule
\end{tabular}
\caption{Complete paired counts, exact two-sided McNemar tests, and percentile
paired-bootstrap 95\% intervals for the right-minus-left accuracy difference.}
\label{tab:paired-stats}
\end{table}

\section{COST-10 Selection and Token Accounting}
\label{app:cost10}
COST-10 is drawn only from the 214 cleaned NL4Opt cases. Cases are sorted by
the Unicode-codepoint count of the unnormalized source text, assigned to ten
equal-frequency bins by $\lfloor 10\,\mathrm{rank}/214\rfloor$, and one case is
drawn per bin with seed string
\texttt{20260726:nl4opt:bin:<k>}. Correctness labels and method outcomes are
not consulted.

\begin{table}[t]
\centering\scriptsize
\begin{tabular}{rrr}
\toprule
Bin & Case & Source chars \\
\midrule
0 & \texttt{prob\_62} & 371 \\
1 & \texttt{prob\_72} & 419 \\
2 & \texttt{prob\_90} & 464 \\
3 & \texttt{prob\_256} & 476 \\
4 & \texttt{prob\_237} & 502 \\
5 & \texttt{prob\_170} & 532 \\
6 & \texttt{prob\_231} & 567 \\
7 & \texttt{prob\_111} & 580 \\
8 & \texttt{prob\_1} & 672 \\
9 & \texttt{prob\_252} & 744 \\
\bottomrule
\end{tabular}
\caption{COST-10 manifest.}
\label{tab:cost-manifest}
\end{table}

All methods receive the same raw descriptions and use the official OpenAI
GPT-4o-2024-08-06 endpoint at temperature 0. Provider-reported prompt and
completion tokens are summed over every successful semantic call. There were no
transient transport retries in the recorded runs. Solver time and one-time
environment setup are excluded. CAFA follows its official one-call generation plus deterministic
cleaning/completion. The intended Chain-of-Experts workflow includes seven
forward experts followed by conductor, reducer, evaluator, and reflection
stages. In the recorded COST-10 run, each case accumulated eight successful
semantic calls before evaluator termination caused by an interface error. The
reported CoE token total is therefore the observed pre-termination cost, not the
cost of a completed end-to-end run, and should be interpreted as a lower bound
on the intended workflow's inference cost. SAC-Opt uses its official iterative
workflow with \texttt{MAX\_ITER=5} and error correction enabled.

\begin{table}[t]
\centering\tiny\setlength{\tabcolsep}{2.5pt}
\begin{tabular}{lrrrr}
\toprule
Method & Input & Output & Total & Calls/inst. \\
\midrule
CAFA & 7,018 & 936 & 7,954 & 1 \\
IR2Solve & 10,068 & 5,963 & 16,031 & 1 \\
CoE & 41,820 & 11,639 & 53,459 & 8 \\
SAC-Opt & 285,827 & 81,201 & 367,028 & 39 \\
\bottomrule
\end{tabular}
\caption{COST-10 token totals and mean observed calls per instance. The CoE
values are observed before evaluator termination and constitute a lower bound
for its intended end-to-end workflow. Relative token volumes are 3.33 (CoE)
and 22.90 (SAC-Opt) times IR2Solve.}
\label{tab:cost-summary}
\end{table}

\begin{table}[t]
\centering\tiny\setlength{\tabcolsep}{4pt}
\begin{tabular}{lrr}
\toprule
Case & CAFA I/O & IR2Solve I/O \\
\midrule
\texttt{prob\_62} & 669/88 & 974/551 \\
\texttt{prob\_72} & 682/87 & 987/580 \\
\texttt{prob\_90} & 689/85 & 994/578 \\
\texttt{prob\_256} & 681/76 & 986/429 \\
\texttt{prob\_237} & 704/103 & 1009/674 \\
\texttt{prob\_170} & 709/108 & 1014/731 \\
\texttt{prob\_231} & 705/96 & 1010/550 \\
\texttt{prob\_111} & 726/109 & 1031/760 \\
\texttt{prob\_1} & 719/88 & 1024/578 \\
\texttt{prob\_252} & 734/96 & 1039/532 \\
\bottomrule
\end{tabular}
\caption{Per-instance input/output tokens for the one-call methods.}
\label{tab:cost-percase-one}
\end{table}

\begin{table}[t]
\centering\tiny\setlength{\tabcolsep}{3pt}
\begin{tabular}{lrrr}
\toprule
Case & CoE I/O & SAC-Opt I/O & SAC calls \\
\midrule
\texttt{prob\_62} & 3771/1033 & 18320/6003 & 30 \\
\texttt{prob\_72} & 4680/1500 & 19199/6098 & 30 \\
\texttt{prob\_90} & 3758/978 & 19145/6101 & 30 \\
\texttt{prob\_256} & 3769/1034 & 16235/6282 & 28 \\
\texttt{prob\_237} & 4104/1156 & 36327/10433 & 48 \\
\texttt{prob\_170} & 4520/1285 & 44083/10696 & 56 \\
\texttt{prob\_231} & 4072/1060 & 37108/9644 & 48 \\
\texttt{prob\_111} & 4322/1114 & 52980/12212 & 60 \\
\texttt{prob\_1} & 4353/1246 & 22081/7054 & 33 \\
\texttt{prob\_252} & 4471/1233 & 20349/6678 & 27 \\
\bottomrule
\end{tabular}
\caption{Per-instance input/output tokens for the multi-call methods. Each
CoE row records eight successful calls before evaluator termination.}
\label{tab:cost-percase-multi}
\end{table}

\FloatBarrier
\section{Datasets and Evaluation Protocol}
\label{app:datasets}
\begin{center}
\scriptsize\setlength{\tabcolsep}{3pt}
\begin{tabular}{lrl}
\toprule
Dataset & $n$ & Role \\
\midrule
NL4Opt & 214 & natural-language LP/MILP formulation \\
IndustryOR & 42 & industry-style problems \\
EasyLP & 545 & easier MAMO subset \\
ComplexLP & 111 & harder MAMO subset \\
NLP4LP & 178 & natural-language-to-LP benchmark \\
ReSocratic & 403 & OptiBench/ReSocratic benchmark \\
\bottomrule
\end{tabular}
\captionof{table}{Survey-cleaned datasets in the main comparison.}
\label{tab:supp-datasets}
\end{center}

Table~\ref{tab:supp-datasets} summarizes the six survey-cleaned datasets used
in the main comparison. The final local evaluation covers NL4Opt, EasyLP,
NLP4LP, and ReSocratic, comprising 1,340 instances and one semantic generation
call per instance. IndustryOR and ComplexLP use the same final ModelIR-based
generation, deterministic verifier, compiler, solver configuration, and
evaluation criterion.

\section{Failure Boundaries and Reproducibility}
\label{app:repro}
Execution outcomes are build failure, built without a recorded objective,
solved with an incorrect objective, and objective-correct. They are not manually
adjudicated semantic labels. The deterministic verifier cannot reconstruct an
omitted constraint or a misread objective, and objective equality does not prove
equality of feasible regions or optimal solution sets.

\paragraph{Experimental configuration.}
The archived experiment metadata specifies GPT-4o-2024-08-06 through the
official OpenAI endpoint; temperature 0; one sample; no provider fallback;
identical-request retries only for transient transport failures; ModelIR
followed by deterministic L1, L2, and compilation; L3, LLM repair, reflection,
and second candidates disabled; and Gurobi with a 60-second limit and otherwise
default parameters. The four-dataset evaluation uses seed 20260728, the
ablation uses seed 20260727, and the bootstrap and COST-10 procedures use seed
20260726. The metadata also includes SHA-256 values for the prompt, schema,
parser, verifier L1/L2, compiler, evaluation, pipeline, and dataset contents.

The archived experiment metadata records the model configuration, prompts,
component hashes, datasets, solver time limit, per-instance outputs, and call
logs, but does not include the original CPU, RAM, operating-system build,
Python version, or exact Gurobi/\texttt{gurobipy} build. We therefore do not
report unverified values for these fields. The implementation requires
\texttt{gurobipy>=10.0.0}.

\clearpage
\bibliographystyle{plainnat}
\bibliography{references}

\end{document}